%% file: main.tex
\documentclass[letter, 10pt, conference]{./cls/ieeeconf}      % Use this line for a4 paper

\IEEEoverridecommandlockouts                              % 
\usepackage{array,multicol,multirow,booktabs,longtable,tabularx}
\usepackage{subcaption}

\usepackage{xspace,algorithmic,algorithm}

\usepackage{amsmath,amsfonts}

\usepackage{amsthm}
\usepackage{amssymb,mathtools,dsfont,bbm,cuted,stmaryrd}

\usepackage{graphicx}      % include this line if your document contains figures

\usepackage{caption}

\usepackage{xcolor}

\usepackage[algo2e, linesnumbered, algoruled]{algorithm2e}

\usepackage{pgfplots}
\usepackage{makecell}

\usepackage{placeins}
\usepackage{cuted}

\let\labelindent\relax
\usepackage{enumitem}

\newcommand{\be}{\begin{equation}}
\newcommand{\ee}{\end{equation}}

\def\bal#1\eal{\begin{align}#1\end{align}}
\def\baln#1\ealn{\begin{align*}#1\end{align*}}

\newcommand{\ben}{\begin{equation*}}
\newcommand{\een}{\end{equation*}}

\newcommand{\re}[1]{\mathbb{R}^{#1}}%

\newcommand{\bbm}{\begin{bmatrix}}
\newcommand{\ebm}{\end{bmatrix}}

\newcommand{\bBm}{\begin{Bmatrix}}
\newcommand{\eBm}{\end{Bmatrix}}

\newcommand{\bvm}{\begin{vmatrix}}
\newcommand{\evm}{\end{vmatrix}}

\newcommand{\bVm}{\begin{Vmatrix}}
\newcommand{\eVm}{\end{Vmatrix}}

\newcommand{\bpm}{\begin{pmatrix}}
\newcommand{\epm}{\end{pmatrix}}

\newcommand{\bnm}{\begin{matrix}}
\newcommand{\enm}{\end{matrix}}

\newcommand{\bi}{\begin{itemize}}
\newcommand{\ei}{\end{itemize}}

\newcommand{\bse}{\begin{subequations}}
\newcommand{\ese}{\end{subequations}}

\DeclareMathOperator{\diag}{diag}
\theoremstyle{plain}

\newtheorem{prop}{Proposition}
\newtheorem{cor}{Corollary}

\newtheorem{df}{Definition}
\newtheorem{rem}{Remark}

\newcommand{\eor}{\ensuremath{\hfill\blacklozenge}}
\def\qed{\hfill\vrule height 1.6ex width 1.5ex depth -.1ex}

\title{\LARGE \bf Inclusion conditions for the Constrained Polynomial Zonotopic case.}

\author{Bogdan Gheorghe$^{1,\star}$, Amr Alanwar$^{2}$ and Florin Stoican$^{1}$% <-this % stops a space
\thanks{$^{1}$Bogdan Gheorghe and Florin Stoican are with Dept. of Automatic Control and Systems Engineering, RePlan team, CAMPUS Research Institute, Politehnica Bucharest, Romania 
        {\tt\small\{bogdan.gheorghe1807, florin.stoican\}@upb.ro}}%
\thanks{$^{2}$Amr Alanwar is with School of Computation, Information and Technology, Technical University of Munich, Heilbronn, 74076, Germany.
        {\tt\small alanwar@tum.de}}%
\thanks{$^\star$ The work of the first author was done while visiting the Technical University of Munich, Campus Heilbronn}%
}

\usepackage{hyperref}

\begin{document}

\maketitle

\begin{abstract}
Set operations are well understood for convex sets but become considerably more challenging in the non-convex case due to the loss of structural properties in their representation. Constrained polynomial zonotopes (CPZs) offer an effective compromise, as they can capture complex, typically non-convex geometries while maintaining an algebraic structure suitable for further manipulation. 

Building on this, we propose novel nonlinear encodings that provide sufficient conditions for testing inclusion between two CPZs and adapt them for seamless integration within optimization frameworks.
\end{abstract}


% \begin{highlights}
% \end{highlights}


\input{chapters/1_introduction}
% \clearpage

% \newpage
\input{chapters/2_preliminaries}

% \clearpage
\input{chapters/3_inclusion_conPolyZon}
\input{chapters/4_results}

\input{chapters/5_conclusions}


\end{document}

%% file: chapters/1_introduction.tex
\section{Introduction}

% Computing the Maximal Positive Invariant set (MPI)

% 
Set-based representations are widely utilized in reachability analysis and fault detection, with notable examples including polytopes~\cite{grunbaum1967convex}, ellipsoids~\cite{kurzhanski2000ellipsoidal}, zonotopes~\cite{kuhn1998zonotope}, and support functions \cite{reachsupport}. Among these, ellipsoids and zonotopes stand out for their computational efficiency in performing fundamental set operations critical to reachability analysis. 

Ellipsoids, defined as affine transformations of the unit 2-norm ball in Euclidean space provide exact representations of Gaussian confidence regions, making them particularly well-suited for modeling stochastic uncertainties. Nevertheless, ellipsoids are not closed under key operations such as the Minkowski sum and intersection, which constrains their use in more complex problems. In contrast, zonotopes possess desirable closure properties under Minkowski sum, linear transformation, and conservative intersection, offering complementary strengths for set-based estimation. There are different variants of zonotopes like hybrid zonotopes~\cite{bird2023hybrid}, logical zonotopes~\cite{alanwar2023logical}, sparse polynomial zonotopes~\cite{kochdumper2020sparse}, and
hybrid polynomial zonotopes~\cite{xie2025hybrid}. Each of these variants has been developed to address specific limitations of conventional zonotopes in representing complex, nonlinear, logical, or hybrid uncertainty structures. 

To exactly describe set intersections, \cite{scott2016constrained} introduced the notion of constrained zonotopes, later extended to the polynomial case in \cite{kochdumper_constrained_2023}. This expressiveness comes at the cost of higher descriptive complexity - namely, an increased number of generators and larger matrix dimensions - making efficient inclusion checking a crucial yet challenging task. Despite significant progress, the problem remains open.  
Linear encodings for polyhedral and zonotopic sets were proposed in~\cite{Sadraddini_Tedrake_2019} and later extended to constrained zonotopes in~\cite{raghuraman_set_2020,Gheorghe_Stoican_Prodan_2024}. Notably, while not exact, in the sense that the test may return false even when inclusion holds, the proposed conditions have been empirically shown to be non-conservative \cite[Sec. IV]{Sadraddini_Tedrake_2019}.

In this context, the present work:
\begin{enumerate}[label=\roman*)]
    \item fills a gap by introducing nonlinear encodings, linear in the equality constraints and nonlinear in the inequalities, for testing constrained polynomial zonotope inclusion;
    \item proposes an adaptation that avoids the explicit use of the absolute value operator, yielding a formulation that remains nonlinear but integrates more naturally within optimization frameworks;
    % \item demonstrates the effectiveness of the proposed methods through the computation of the maximal invariant set for a discrete-time linear dynamical system, a demanding task involving repeated matrix multiplications, set intersections, and inclusion checks.
\end{enumerate}

% \vfill

The rest of the paper is organized as follows. Section~\ref{sec:preliminaries} introduces the notion of constrained polynomial zonotope. Section~\ref{sec:inclusion_cpz} defines the sufficient conditions for CPZ inclusion. The validation of the theoretical results are shown in Section~\ref{sec:results}, while Section~\ref{sec:conclusions} draws the conclusions.

\subsection*{Notations}
% \noindent\emph{Notation}:
$\mathbb N_0$ is the set of whole numbers. $O_{m \times n}\in \mathbb R^{m\times n}$ is the matrix with $m$ rows and $n$ columns whose entries are zero. Whenever $m=n$, we use the shorthand notation $O_{n}$. The identity matrix is denoted by $I_n\in \mathbb R^{n\times n}$, and the symbol $\mathbf{1}_n$ represents the column vector of $n$ values of one. For an arbitrary matrix $G\in \mathbb R^{m\times n}$, $G^\top \in \mathbb R^{n\times m}$ is its transpose matrix, $G^\dagger \in \mathbb R^{n\times m}$ its Moore-Penrose pseudoinverse, $G(i, \cdot) \in \mathbb R^{1\times n}$ its i-th row, and $G(\cdot, j)  \in \mathbb R^{m\times 1}$ its j-th column. $|A|$, with $A \in \mathbb R^{m\times n}$, represents the element-wise absolute value of $A$. For a vector $x \in \mathbb R^{n}$, its $1$ norm is given as $\|x\|_{1} := \sum_{i=1}^{i=n}|x_i|$ and its infinity norm as $\|x\|_{\infty} := \max(|x_1|, |x_2|, \ldots, |x_n|)$. Any order operations (e.g., ``$\leq$'' or ``$\geq$'') or operator (e.g., ``$\log (\cdot)$'') are taken element-wise over the vector(s) considered. The function $\diag(v)$ constructs a square matrix, with the elements of the vector $v$ on its main diagonal and zeros elsewhere. % The function \textit{diag($v$)} creates the square matrix $G$ that has on its diagonal the vector $v$, and zero otherwise.For two sets, $X$ and $Y$, their Minkowski sum is defined as $X\oplus Y=\{x+y:\: \forall x \in X, \forall y\in Y\}$, and their Pontryagin difference is defined as $X\ominus Y=\{x\in X:\: x + y \in X, \forall \, y \in Y\}$.

%% file: chapters/2_preliminaries.tex
\section{Preliminaries on zonotopes}
\label{sec:preliminaries}
% \subsection{Data-driven methods}

% \newpage

% \subsection{Zonotopic representations}
% To do:
% \begin{itemize}
%     \item \textcolor{red}{Do the figure with lambda}
% \end{itemize}

We introduce the notion of constrained polynomial zonotope \cite{kochdumper_constrained_2023} and its subsequent simplifications, to characterize the sets appearing later in the paper. 

\begin{df}
A compact set $\mathcal{P} \subset \re{d}$ is a constrained polynomial zonotope (CPZ) if there exists $c \in \mathbb R^{d}, G \in \mathbb R^{d \times n}, E \in \mathbb N_0^{s \times n}, F \in \mathbb \re^{p\times q}, \theta \in \mathbb R^{p}, R \in \mathbb N_0^{s \times q}$ such that
\begin{align}
\label{eq:con_poly_zono}
    \nonumber\mathcal{P} = & \bigl\langle c, G, E, F, \theta, R\bigr\rangle_{CPZ}\\
    =&\bigg \{ x\in \mathbb R^{d}:\: x = c + \sum_{i=1}^{n} \bigg( \prod_{k=1}^s \lambda^{E{(k,i)}}(k)\bigg) G(\cdot, i), \nonumber\\ 
    &\sum_{i=1}^{q}\bigg(\prod_{k=1}^s \lambda^{R(k,i)}(k) \bigg) F(\cdot, i) = \theta, \|\lambda\|_\infty \leq 1\bigg \}.
\end{align}
\end{df}
\noindent A couple of remarks are in order.
\begin{rem}
Notably, the polynomial component can represent non-convex shapes. This capability arises from the nonlinearities introduced by the polynomial terms and from the decoupling between the number of generators ($n$), the number of equality generators ($q$), and the dimension of the $\lambda$-space ($s$).\eor
\end{rem}
\begin{rem}
A polynomial zonotope (PZ) is a special case of a CPZ with no equality constraints, i.e., $F$ and $\theta$ are empty.  A constrained zonotope (CZ) is obtained when the polynomial component is linear, meaning $E$ and $R$ are identity matrices.  Finally, a zonotope (Z) is a particular case of a CZ with no equality constraints, that is, $F$ and $\theta$ are empty. \eor
\end{rem}
CPZ sets are closed under matrix multiplication and intersection. For the former, with the notation from \eqref{eq:con_poly_zono} and for some matrix $M\in \mathbb R^{m\times d}$ we have that 
\begin{equation}
\label{eq:cpz_matrix_mul}
    M\mathcal P = \langle Mc, MG, E, F, \theta, R\rangle_{CPZ}.
\end{equation}
% For the intersection case, consider the CPZ sets 
% $\mathcal{P}_i = \langle c_i, G_i, E_i, F_i, \theta_i, R_i\rangle_{CPZ} \in (\mathbb R^{d}, \mathbb R^{d \times n_i}, \mathbb N_0^{s_i \times n_i}, \mathbb \re^{p_i\times q_i}$, $\mathbb R^{p_i}, \mathbb N_0^{s_i \times q_i})$ for $i \in \{1, 2\}$, then their intersection is \cite[Proposition 3.2.23]{Kochdumper}
% \begin{multline}
% \label{eq:cpz_intersection}
% \mathcal{P}_1 \cap \mathcal{P}_2 = \bigg\langle
% c_1,\, G_1,\,
% \begin{bmatrix}
% E_1 \\[3pt]
% \mathbf{O}
% \end{bmatrix},\,
% \begin{bmatrix}
% F_1 & \mathbf{O} & \mathbf{O} & \hphantom{-}\mathbf{O} \\[3pt]
% \mathbf{O} & F_2 & \mathbf{O} & \hphantom{-}\mathbf{O} \\[3pt]
% \mathbf{O} & \mathbf{O} & G_1 & -G_2
% \end{bmatrix},\\
% \begin{bmatrix}
% \theta_1 \\[3pt]
% \theta_2 \\[3pt]
% c_2 - c_1
% \end{bmatrix},\,
% \begin{bmatrix}
% R_1 & \mathbf{O} & E_1 & \mathbf{O} \\[3pt]
% \mathbf{O} & R_2 & \mathbf{O} & E_2
% \end{bmatrix}
% \bigg\rangle_{CPZ},
% \end{multline}
% where $\mathbf{O}$ is the zero entry matrix of appropriate dimension.

\subsection*{Illustrative example}
\label{subsec:CPZ_example}
Consider the $\mathcal{CPZ}$ example (taken from \cite{kochdumper_constrained_2023}, slightly modified for illustrative purposes):
\begin{align*}
\mathcal{P} \;=\;
\bigg\langle
&\begin{bmatrix}0\\[4pt]0\end{bmatrix},
\;
\begin{bmatrix}
1 & 0 & 1 & -1\\[4pt]
0 & 1 & 1 & 1
\end{bmatrix},
\;
\begin{bmatrix}
1 & 0 & 1 & 2\\[4pt]
0 & 1 & 1 & 0\\[4pt]
0 & 0 & 1 & 1
\end{bmatrix},\\
&\qquad\qquad\qquad\quad[\,1\; 1\; 1\,],
\;
1.5,\;
\begin{bmatrix}
0 & 1 & 2\\[4pt]
1 & 0 & 0\\[4pt]
0 & 1 & 0
\end{bmatrix}
\bigg\rangle_{\text{CPZ}}\\
=
\bigg\{\mkern-6mu
&\begin{bmatrix}0\\[4pt]0\end{bmatrix}
\mkern-4mu+\mkern-4mu
\begin{bmatrix}1\\[4pt]0\end{bmatrix}\lambda_1
\mkern-4mu+\mkern-4mu
\begin{bmatrix}0\\[4pt]1\end{bmatrix}\lambda_2
\mkern-4mu+\mkern-4mu
\begin{bmatrix}1\\[4pt]1\end{bmatrix}\lambda_1\lambda_2\lambda_3
\mkern-4mu+\mkern-4mu
\begin{bmatrix}-1\\[4pt]\hphantom{-}1\end{bmatrix}\lambda_1^2\lambda_3
\;\hspace{-0.3em},\\
% &\begin{array}{l}
% \lambda_2 - 0.5\,\lambda_1\lambda_3 + 0.5\,\lambda_1^2 = 0.5,\\[4pt]
% \|\lambda_i\|_\infty \leq 1, \forall i \in \{1, 2, 3\}
% \end{array}
&\qquad\qquad\quad\lambda_2 + \,\lambda_1\lambda_3 + \,\lambda_1^2 = 1.5, \|\lambda\|_\infty \leq 1
\bigg\}.
\end{align*}
Figure~\ref{fig:different_zon} illustrates the original CPZ~$\mathcal{P}$ (solid red) alongside its simplified representations.  
The solid blue contour corresponds to the PZ case, obtained by ignoring the equality constraints.  The CZ and Z variants share the same generators but omit the polynomial component.  For the CZ case, the equality constraint is $\lambda_1 + \lambda_2 + \lambda_3 + 0 \cdot \lambda_4 = 1.5$, whereas in the Z case it is omitted entirely.  
Figure~\ref{fig:constr_zon} further shows the effect of constraining~$\lambda$: the hypercube domain is intersected with the polynomial (or hyperplane) constraint, yielding the CPZ and CZ representations, respectively.

% \begin{figure}[!ht]
% \centering
% \includegraphics[width=\columnwidth]{pics/CPZ_illustrative_example/illustrative_example_all.pdf}
% \caption{Example different zonotopic representations}
% \label{fig:different_zon}
% \end{figure}
\begin{figure}[!ht]
\centering
\subfloat[representation in $\mathbb R^2$]{\label{fig:different_zon}\includegraphics[width=.825\columnwidth]{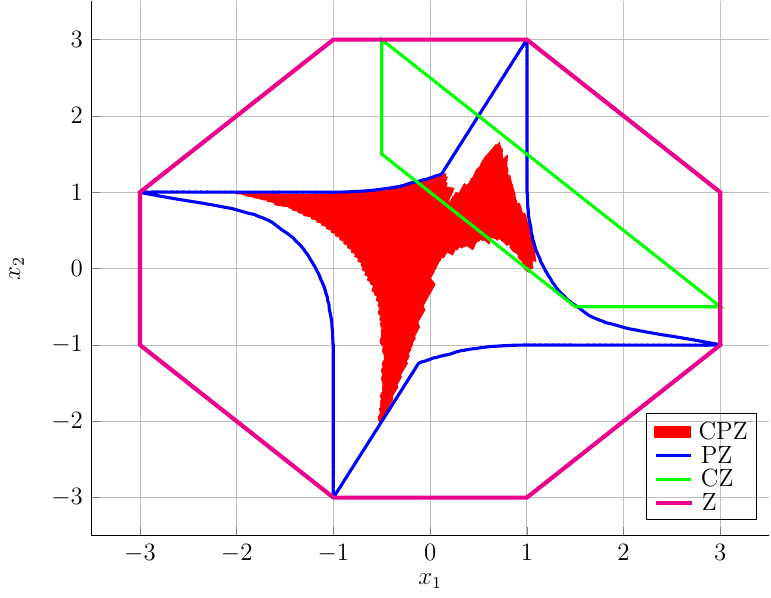}}\\
\subfloat[constraints for $\lambda$ (in $\mathbb R^3$)]{\label{fig:constr_zon}\includegraphics[width=.825\columnwidth]{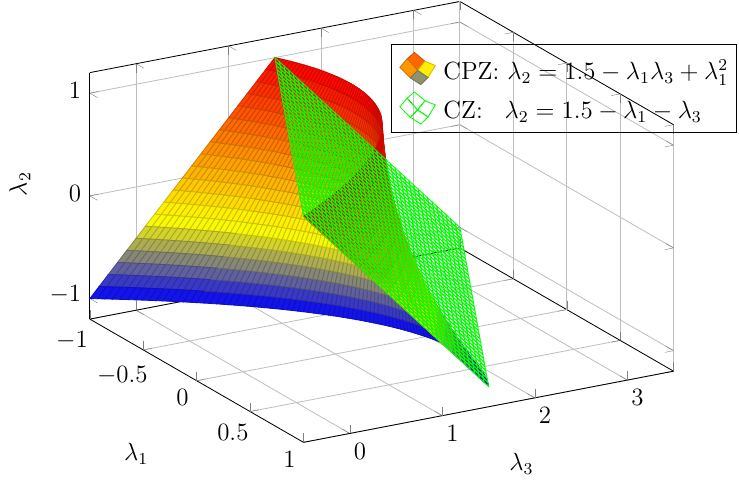}}
\caption{Different zonotopic representations}
\label{fig:mpi_bar_kq}
\end{figure}

% Fig.~\ref{fig:different_zon} showcases,  $\mathcal P$, the original CPZ (solid red) as well as simpler representations: solid blue denotes the PZ case (where the equality constraints are ignored). The CZ and Z use the same generators, without the polynomial component. The equality is taken as $\lambda_1 + \lambda_2 + \lambda_3 + 0 \cdot \lambda_4= 1.5$ for the CZ case and ignored in the Z case. Fig.~\ref{fig:constr_zon} illustrates the effect of constraining $\lambda$ (the hypercube inclusion, later intersected with the polynomial/hyperplane constraint that corresponds to the CPZ and CZ cases, respectively).

%% file: chapters/3_inclusion_conPolyZon.tex
\section{Inclusion conditions for constrained polynomial zonotopes}
\label{sec:inclusion_cpz}
Many control application, the computation of the MPI set included, require set inclusion checks. To this end, we first extend the results of~\cite{Sadraddini_Tedrake_2019} and our previous work~\cite{10552328} to derive sufficient conditions for CPZ inclusion.

\begin{prop}
\label{prop:cpz_inclusion}
Consider the constrained polynomial zonotopic sets 
$\mathcal{P}_i = \langle c_i, G_i, E_i, F_i, \theta_i, R_i\rangle_{CPZ} \in (\mathbb R^{d}, \mathbb R^{d \times n_i}, \mathbb N_0^{s_i \times n_i}, \mathbb \re^{p_i\times q_i}, \mathbb R^{p_i}, \mathbb N_0^{s_i \times q_i})$ for $i \in \{1, 2\}$, and matrices $\gamma \in \mathbb{R}^{n_2}$, $\Gamma \in \mathbb{R}^{n_2 \times n_1}$, $\Pi \in \mathbb{R}^{p_2 \times p_1}$, $\Psi \in \mathbb{R}^{q_2 \times q_1}$, $\psi \in \mathbb{R}^{q_2}$. Then, the conditions
\begin{subequations}
\label{eq:cpz_inclusion}
    \begin{align}
        \label{eq:cpz_inclusion_a}c_1 &= c_2+G_2\gamma,\\
        \label{eq:cpz_inclusion_b}G_1 &=G_2\Gamma,\\
        \label{eq:cpz_inclusion_c}\Pi F_1& =F_2\Psi,\\
        \label{eq:cpz_inclusion_d}\Pi \theta_1& =\theta_2-F_2\psi,\\
        % \label{eq:cpz_inclusion_e}\bbm E_2^\top\\[5pt]R_2^\top\ebm^\dagger\cdot \log \bbm \left|\gamma\right| + \left|\Gamma\right| \cdot \mathbf 1_{n_1}\\[5pt] \left|\psi\right| + \left|\Psi\right| \cdot \mathbf 1_{q_1}\ebm &\leq O_{2s_2 \times 1},
        \label{eq:cpz_inclusion_e} (E_2^\top)^\dagger \cdot \log\left( \left|\gamma\right| + \left|\Gamma\right| \cdot \mathbf 1_{n_1}\right) &\leq O_{s_2 \times 1},\\
        \label{eq:cpz_inclusion_f} (R_2^\top)^\dagger \cdot \log\left( \left|\psi\right| + \left|\Psi\right| \cdot \mathbf 1_{q_1}\right) &\leq O_{s_2 \times 1},
    \end{align}
\end{subequations}    
are sufficient to guarantee the inclusion $\mathcal{P}_1 \subseteq \mathcal{P}_2$.
\end{prop}

\begin{proof}
To prove that $\mathcal{P}_1 \subseteq \mathcal{P}_2$, we consider an arbitrary element $x \in \mathcal{P}_1$ and show, using~\eqref{eq:cpz_inclusion}, that $x \in \mathcal{P}_2$. Equivalently stated, for any admissible $\lambda_1$ there exists a feasible $\lambda_2$ such that the inclusion holds:
\begin{equation}
\label{eq:lambda_inclusion}
    \forall \lambda_1 \text{ s.t. }x \in \mathcal P_1, \exists \lambda_2\text{ s.t. } x \in \mathcal P_2.
\end{equation}
\noindent The following chain of inferences holds:
\begin{align*}
    x &= c_1 + \sum_{j=1}^{n_1} \left( \prod_{k=1}^{s_1} \lambda_{1}^{E_1(k,j)}(k) G_{1}(\cdot, j) \right) \\ 
    &\overset{\eqref{eq:cpz_inclusion_a}+\eqref{eq:cpz_inclusion_b}}{=}  
      c_2 + G_2 \gamma + \sum_{j=1}^{n_1} \left( \prod_{k=1}^{s_1} \lambda_{1}^{E_1(k,j)}(k) G_2 \Gamma(\cdot, j)\right) \\
      &= c_2 + \sum_{i=1}^{n_2}
      \left[\mkern-4mu \biggl( \gamma(i) +\mkern-4mu \sum_{j=1}^{n_1}\mkern-4mu 
      \biggl(\Gamma(i, j)\mkern-4mu \prod_{k=1}^{s_1} \lambda_{1}^{E_1(k,j)}(k) \biggr)\mkern-4mu\biggr) G_{2}(\cdot, i) \right].
\end{align*}  
Hence, and using \eqref{eq:lambda_inclusion}, we require that $\lambda_2$ verifies
\begin{equation}
\label{eq:equiv_lambda_cpz_gen}
    \gamma(i) + \sum_{j=1}^{n_1} 
    \biggl(\Gamma(i, j) \prod_{k=1}^{s_1} \lambda_{1}^{E_1(k,j)}(k) \biggr)
    = \prod_{k=1}^{s_2} \lambda_{2}^{E_2(k,i)}(k),
\end{equation}
for all $i \in \{1,2,\ldots,n_2\}$. Next, by left multiplying with $\Pi$ the equalities constraining $\lambda_1$ as in \eqref{eq:con_poly_zono}, they become
\begin{equation}
\label{eq:equiv_lambda_cpz2}
    \sum_{j=1}^{n_1}\bigg(\prod_{k=1}^{s_1} \lambda_{1}^{R_1{(k,j)}}(k) \bigg) \Pi F_{1}(\cdot, j) = \Pi \theta_1.
\end{equation}
Applying~\eqref{eq:cpz_inclusion_c} to the left-hand side (\textit{lhs}) of~\eqref{eq:equiv_lambda_cpz2}, we get
\begin{align*}
    \textit{lhs} &= \sum_{j=1}^{q_1}\bigg(\prod_{k=1}^{s_1} \lambda_{1}^{R_1{(k,j)}}(k) \bigg) F_{2}\Psi(\cdot, j) \\
    &= \sum_{i=1}^{q_2}\left[\sum_{j=1}^{q_1}\Psi(i, j)\bigg(\prod_{k=1}^{s_1} \lambda_{1}^{R_1{(k,j)}}(k) \bigg)F_2(\cdot, i)\right],
\end{align*}

\noindent which, combined with the application of~\eqref{eq:cpz_inclusion_d} in the right-hand side (\textit{rhs}) of~\eqref{eq:equiv_lambda_cpz2}, leads to
\begin{equation*}
    \sum_{i=1}^{q_2} \bigg[\psi(i) + \sum_{j=1}^{q_1}\Psi(i, j)\bigg(\prod_{k=1}^{s_1} \lambda_{1}^{R_1{(k,j)}}(k) \bigg)\bigg]F_2(\cdot, i) = \theta_2.
\end{equation*}
Similar to~\eqref{eq:equiv_lambda_cpz_gen}, this relation leads to a constraint on $\lambda_2$:
\begin{equation}
\label{eq:equiv_lambda_cpz_equal}
    \psi(i) +\mkern-4mu \sum_{j=1}^{q_1}\Psi(i, j)\mkern-4mu\bigg(\prod_{k=1}^{s_1}\mkern-4mu \lambda_{1}^{R_1{(k,j)}}(k)\bigg) = \prod_{k=1}^{s_2} \lambda_{2}^{R_2{(k,i)}}(k) 
\end{equation}
for all $i \in \{1,2,\dots,q_2\}$.

For convenience, denote the \textit{lhs} of~\eqref{eq:equiv_lambda_cpz_gen} with $\mu_{i_{1}}^E$, and the \textit{lhs} of~\eqref{eq:equiv_lambda_cpz_equal} with $\mu_{i_{2}}^R$, and we require\footnote{We replaced $i$ in the \textit{lhs} of both equations with $i_{1}$, and $i_{2}$ respectively.}:
\begin{equation}
\label{eq:mu_ER}
    \mu_{i_{1}}^E = \prod_{k=1}^{s_2} \lambda_{2}^{E_2(k,i_{1})}(k),\quad \mu_{i_{2}}^R = \prod_{k=1}^{s_2} \lambda_{2}^{R_2{(k,i_{2})}}(k).
\end{equation}
Applying the absolute value and logarithm operators successively in \eqref{eq:mu_ER} yields
\begin{align*}
    \log\left|\mu_{i_{1}}^E\right| &= \sum_{k=1}^{s_2} E_2(k,i_{1})\log|\lambda_{2}(k)| = E_2^\top(\cdot, i_{1}) \log|\lambda_2|,\\
    \log\left|\mu_{i_{2}}^R\right| &= \sum_{k=1}^{s_2} R_2(k,i_{2})\log|\lambda_{2}(k)| = R_2^\top(\cdot, i_{2}) \log|\lambda_2|,
\end{align*}
for all $i_{1} \in \{1,2,\dots,n_2\}$, and $i_{2} \in \{1,2,\dots,q_2\}$. Repeating for all indices $i_{1}$, and $i_{2}$ gives% and stacking the resulting expressions gives\footnote{Column vectors $\mu^E, \mu^R$ stack vertically scalars $\mu^E_i$, $\mu^R_i$ when iterating after $i$.} 
% \begin{equation}
%     \label{eq:log_mu_ER}
%     \bbm \log\left|\mu^E\right|\\[5pt]\log\left|\mu^R\right|\ebm = \bbm E_2^\top\\[5pt]R_2^\top\ebm \log|\lambda_2|.
% \end{equation}
\begin{equation}
    \label{eq:log_mu_ER}
    \log\left|\mu^E\right| = E_2^\top\log|\lambda_2|, \quad \log\left|\mu^R\right|=R_2^\top\log|\lambda_2|.
    % \mu_{i_{1}}^E = \prod_{k=1}^{s_2} \lambda_{2}^{E_2(k,i)}(k),\quad \mu_{i_{2}}^R = \prod_{k=1}^{s_2} \lambda_{2}^{R_2{(k,j)}}(k).
\end{equation}

Recalling that $\mu_{i_{1}}^E$ is the \textit{lhs} of~\eqref{eq:equiv_lambda_cpz_gen} and that it must hold for any $\|\lambda_1\|_\infty \le 1$, the triangle inequality gives 
\begin{align}
\label{eq:abs_mu_i}
    \bigl|\mu_{i_{1}}^E\big| &= \bigg|\gamma(i_{1}) + \sum_{j=1}^{n_1} \bigg( \Gamma(i_{1}, j) \prod_{k=1}^{s_1} \lambda_{1}^{E_1(k,j)}(k) \bigg) \bigg| \nonumber \\
    &\leq  \big|\gamma(i_{1})\big| + \bigg|\sum_{j=1}^{n_1} \bigg( \Gamma(i_{1}, j) \prod_{k=1}^{s_1} \lambda_{1}^{E_1(k,j)}(k) \bigg)\bigg| \nonumber \\
    &\leq \big|\gamma(i_{1})\big| +\sum_{j=1}^{n_1}  \left|\Gamma(i_{1}, j)\right|\cdot \underbrace{\bigg|\prod_{k=1}^{s_1} \lambda_{1}^{E_1(k,j)}(k)\bigg|}_{\leq 1}\nonumber \\[-10pt]
    &\leq \big|\gamma(i_{1})\big| + \sum_{j=1}^{n_1} \big|\Gamma(i_{1}, j)\big|.
\end{align}

Repeating \eqref{eq:abs_mu_i} for $\mu_{i_{2}}^R$, the \textit{lhs} of \eqref{eq:equiv_lambda_cpz_equal}, and for all indices $i_{1}\in \{1, \ldots, n_2\}$, $i_{2} \in \{1, \ldots, q_2\}$ we get:
\begin{equation}
\label{eq:abs_mu_E_R_eq}
    \left|\mu^E\right| \leq \left|\gamma\right| + \left|\Gamma\right| \cdot \mathbf 1_{n_1}, \quad
    \left|\mu^R\right| \leq \left|\psi\right| + \left|\Psi\right| \cdot \mathbf 1_{q_1}.
\end{equation}
Using \eqref{eq:abs_mu_E_R_eq} in \eqref{eq:log_mu_ER}, with the assumption of full-rank pseudo-inverses for $E_2^\top, R_2^\top$ and exploiting the monotonicity of the $\log$ operation, we have that
% \begin{equation}
%     \log|\lambda_2|
%     \leq \bbm E_2^\top\\[5pt]R_2^\top\ebm^\dagger\cdot \log \bbm \left|\gamma\right| + \left|\Gamma\right| \cdot \mathbf 1_{n_1}\\[5pt] \left|\psi\right| + \left|\Psi\right| \cdot \mathbf 1_{q_1}\ebm,
% \end{equation}
\begin{align}
    \log|\lambda_2| &\leq (E_2^\top)^\dagger \cdot \log\left( \left|\gamma\right| + \left|\Gamma\right| \cdot \mathbf 1_{n_1}\right),\\
    \log|\lambda_2| &\leq (R_2^\top)^\dagger \cdot \log \left( \left|\psi\right| + \left|\Psi\right| \cdot \mathbf 1_{q_1}\right).
\end{align}
which allows to state that \eqref{eq:cpz_inclusion_e} and \eqref{eq:cpz_inclusion_f} represent a sufficient condition for $\log|\lambda_2|\leq 0$ which implies that $|\lambda_2|\leq 1$ therefore $x\in \mathcal P_2$, thus concluding the proof. \qed
\end{proof}
Prop.~\ref{prop:cpz_inclusion} is next adapted to avoid the explicit use of the absolute operator in \eqref{eq:cpz_inclusion_e}--\eqref{eq:cpz_inclusion_f}, similar to \cite[Section III.B]{Gheorghe_Stoican_Prodan_2024}.
\begin{cor}
\label{cor:alpha_adaptation}
    With the notation of Prop.~\ref{prop:cpz_inclusion}, adding the auxiliary variables $\alpha_\Gamma\in \mathbb R^{2(n_1+1) \times n_2}$ and $\alpha_\Psi\in \mathbb R^{2(q_1+1) \times q_2}$ allows to reformulate \eqref{eq:cpz_inclusion_e}--\eqref{eq:cpz_inclusion_f} into the equivalent form:
        \begin{multline}
        \label{eq:alpha_gamma}
        (E_2^\top)^\dagger \cdot \log\left(\alpha_\Gamma^\top \mathbf 1_{2(n_1+1)}\right) \leq O_{s_2 \times 1},\\
        \bbm \Gamma & \gamma\ebm^\top=\bbm \mathrm I_{n_1+1}&-\mathrm I_{n_1+1}\ebm\alpha_\Gamma,\quad
        \alpha_{\Gamma}\geq 0,
    \end{multline}
and
    \begin{multline}
    \label{eq:alpha_psi}
        (R_2^\top)^\dagger \cdot \log\left(\alpha_\Psi^\top \mathbf 1_{2(q_1+1)}\right) \leq O_{s_2 \times 1},\\
        \bbm \Psi & \psi\ebm^\top=\bbm \mathrm I_{q_1+1}&-\mathrm I_{q_1+1}\ebm\alpha_\Psi,\quad
        \alpha_{\Psi}\geq 0.
    \end{multline}
\end{cor}
\begin{proof}
Consider the $i$-th row from the logarithmic argument in~\eqref{eq:cpz_inclusion_e}. We have 
\[
\left|\gamma_i\right| + \left|\Gamma(i,\cdot)\right|\mathbf{1}_{n_1} 
= \bigl\|\begin{bmatrix}\gamma_i & \Gamma(i,\cdot)\end{bmatrix}\bigr\|_1 
= \|\zeta_i\|_1.
\]
Hence, verifying $\|\zeta_i\|_1 \le \nu$ is equivalent to checking whether 
$\zeta_i \in \nu CR_{n_1+1}$, i.e., whether it lies within the scaled unit cross-polytope in $\mathbb{R}^{n_1+1}$ (the dual of the unit hypercube, \cite[Sec. 5]{fukuda_polyhedral_2020}). 
The hypercube admits $2^{n_1+1}$ supporting hyperplanes, requiring an equal number of inequalities to verify the inclusion.  
However, its vertex-based representation contains only $2(n_1+1)$ vertices.  
This observation allows $\zeta_i$ to be expressed as a convex combination of these vertices:
\[
\zeta_i = 
\begin{bmatrix} 
\mathrm{I}_{n_1+1} & -\mathrm{I}_{n_1+1} 
\end{bmatrix}
\alpha_\Gamma(i,\cdot),
\]
where $\alpha_\Gamma^\top(i,\cdot)\mathbf{1}_{2(n_1+1)} \le \nu$ and $\alpha_\Gamma(i,\cdot) \ge 0$.  
Repeating this reasoning for each index $i$ yields~\eqref{eq:alpha_gamma}, while applying the same argument to the logarithmic term in~\eqref{eq:cpz_inclusion_f} gives~\eqref{eq:alpha_psi}, thereby concluding the proof.
\qed
\end{proof}
\begin{rem}
Both~\eqref{eq:cpz_inclusion} from Proposition~\ref{prop:cpz_inclusion} and its adaptation~\eqref{eq:alpha_gamma}--\eqref{eq:alpha_psi} in Corollary~\ref{cor:alpha_adaptation} are nonlinear and appear within larger optimization problems.  
It is thus relevant to assess their respective problem sizes in terms of variables, equalities, and inequalities:
\begin{center}
\scriptsize
\setlength\tabcolsep{2.5pt}
\renewcommand*{\arraystretch}{1.25}
\begin{tabular}{c|p{2.75cm}|c|c}
case & ~~~~~~~\# vars. & \# eqs. & \# ineqs.\\\hline
Prop.~\ref{prop:cpz_inclusion} & \makecell{$n_2(n_1+1)+$\\$p_2p_1+q_2(q_1+1)$} & \makecell{$d(n_1+1)+$\\$p_2(q_1+1)$}& $2s_2$\\\hline
Cor.~\ref{cor:alpha_adaptation} & \makecell{$3n_2(n_1+1)+$\\$p_2p_1+3q_2(q_1+1)$} & \makecell{$(d+n_2)(n_1+1) +$\\$(p_2+q_2)(q_1+1)$} & \makecell{$2[s_2+n_2(n_1+1)+$\\$q_2(q_1+1)]$} \\
\end{tabular}
\setlength\tabcolsep{6pt}
\end{center}

Although Corollary~\ref{cor:alpha_adaptation} introduces slightly more variables, it benefits from simpler nonlinearities, involving only the logarithmic operator.
% Both \eqref{eq:cpz_inclusion} from Prop.~\ref{prop:cpz_inclusion} and the adaptation \eqref{eq:alpha_gamma}--\eqref{eq:alpha_psi} given in Cor.\ref{cor:alpha_adaptation} are both nonlinear and embedded in larger optimization problems. It is then worthwhile to asses the number of variables, equalities and inequalities:
% \begin{center}
% \small
% \setlength\tabcolsep{2.5pt} % default value: 6pt
% \renewcommand*{\arraystretch}{1.25}
%     \begin{tabular}{c|p{2.75cm}|c|c}
%     case & ~~~~~~~\# vars. & \# eqs. & \# ineqs.\\\hline
%     Prop.\ref{prop:cpz_inclusion} & $n_2(n_1 + 1) + p_2p_1 + q_2(q_1 + 1)$ & $d(n_1 + 1) + p_2(q_1 + 1)$ & $2s_2$\\\hline
%     Cor.\ref{cor:alpha_adaptation} & $1$ & $2$ & $3$\\
%     \end{tabular}
% \setlength\tabcolsep{6pt}
% \end{center}
% While Cor.~\ref{cor:alpha_adaptation} is larger but it counterbalances by the simpler nonlinearities involved (only the $\log$ operator). 
\eor
\end{rem}
Prop.~\ref{prop:cpz_inclusion} gives a nonlinear encoding for CPZ inclusion.  Previous results such as \cite[Prop.2]{10552328} handle the CZ case. The next result shows how the later may be obtained as a simplification of the former.

\begin{cor}
\label{cor:inclusion}
When $\mathcal P_1$, $\mathcal P_2$ reduce to constrained zonotopes, the inclusion test \eqref{eq:cpz_inclusion} reduces to
\begin{subequations}
\label{eq:cz_inclusion}
    \begin{align}
        \label{eq:cz_inclusion_a}c_1 &= c_2+G_2\gamma,
        &G_1&=G_2\Gamma\\
        \label{eq:cz_inclusion_c}\Pi F_1& =F_2\Gamma, &\Pi \theta_1 &=\theta_2-F_2\gamma,\\
        \label{eq:cz_inclusion_e}  \left|\gamma\right| + \left|\Gamma\right| \cdot \mathbf 1_{n_1} &\leq \mathbf 1_{s_2 \times 1}.
    \end{align}
\end{subequations}
\end{cor}
\begin{proof}
Taking $E_1 = R_1 = I_{n_1}$ and $E_2 = R_2 = I_{n_2}$ for the CPZ sets from Prop.~\ref{prop:cpz_inclusion} reduces them to the CZ case. As a first effect, \eqref{eq:cpz_inclusion_e}--\eqref{eq:cpz_inclusion_f} have the equivalent linear forms
\begin{equation}
\label{eq:pi}
     \left|\gamma\right| + \left|\Gamma\right| \cdot \mathbf 1_{n_1}\leq \mathbf 1_{s_2}, \quad  \left|\psi\right| + \left|\Psi\right| \cdot \mathbf 1_{q_1}\leq \mathbf 1_{s_2}.
\end{equation}
In addition, the dimensions are now linked by $n_1=q_1=s_1$ and $n_2=q_2=s_2$, respectively.
Solving for \eqref{eq:cpz_inclusion_a}, \eqref{eq:cpz_inclusion_b} and \eqref{eq:cpz_inclusion_e} gives $\Gamma$, $\gamma$. Next, we observe that \eqref{eq:cpz_inclusion_c}--\eqref{eq:cpz_inclusion_d} may be used to obtain 
\begin{equation*}
    \Pi=\bbm F_2\Psi\\ \theta_2-F_2\psi\ebm\cdot \bbm F_1\\\theta_1\ebm^\dagger
\end{equation*}
where matrix $\Pi \in \mathbb R^{p_2\times p_1}$ is obtained as a function of $\Gamma,\gamma$, always possible since $\bbm F_1\\\theta_1\ebm\in \mathbb R^{(p_1+1)\times s_1}$ which admits\footnote{Otherwise, since $q_1=s_1$, the constraints over-determine the system and no solution exists.} a full-rank pseudo-inverse since $p_1+1\leq q_1=s_1$. Thus, by choosing $\Psi=\Gamma$ and $\psi=\gamma$ in \eqref{eq:pi} we have that a matrix $\Pi$ exists and that \eqref{eq:cpz_inclusion_f} reduces to \eqref{eq:cpz_inclusion_e}, thus concluding the proof. \qed
\end{proof}
\begin{rem}
    One may ask, conversely, why the variables $\Psi$ and $\psi$ are needed in the first place. The reason lies in the possible mismatch between the numbers of columns in $G_{1,2}$ and $F_{1,2}$, which are decoupled from the dimensions of the corresponding vectors $\lambda_{1,2}$. Simply put, $\Gamma$ and $\Psi$ may not have the same dimensions.\eor
\end{rem}

%% file: chapters/4_results.tex
\section{Simulation results for constrained polynomial zonotopes inclusion}
\label{sec:results}

We consider three CPZ sets defined as
\begin{equation*}
    \mathcal{P}_i=\langle c, G \cdot \delta_{G,i}, E, F \cdot \delta_{F,i}, \theta, R \rangle_{CPZ}, \; i\in \{1,2,3\},
\end{equation*}
where the matrices $c, G, E, F, \theta$, and $R$ are taken from the previous illustrative example, and the scaling factors $\delta_{G,i}$ and $\delta_{F,i}$ are given by
\begin{center}
\renewcommand{\arraystretch}{1.25}
\setlength\tabcolsep{3pt}
\begin{tabular}{c|c|c}
Index $i$ & $\delta_{G,i}$ & $\delta_{F,i}$ \\\hline
1 & $\diag(0.9,0.9,0.72,0.72)$ & $\diag(0.9,0.81,0.81)$ \\\hline
2 & $\diag(1,1,1,1)$ & $\diag(1,1,1)$ \\\hline
3 & $\diag(1.18,1.18,1.64,1.64)$ & $\diag(1.18,1.39,1.39)$ \\
\end{tabular}
\setlength\tabcolsep{6pt}
\end{center}
The sets are illustrated in Fig.~\ref{fig:CPZ_inclusion_test}, where $\mathcal{P}_1$ is shown with a blue contour, $\mathcal{P}_2$ as filled red, and $\mathcal{P}_3$ with a green contour.
\begin{figure}[!ht]
\centering
\includegraphics[width=\columnwidth]{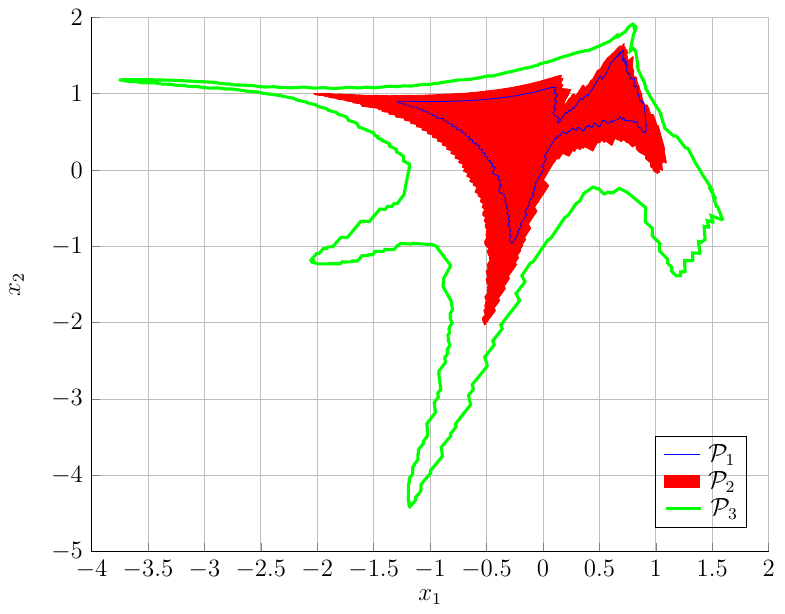}
\caption{CPZ sets used for the inclusion validation.}
\label{fig:CPZ_inclusion_test}
\end{figure}

We test each pair in both directions (i.e., $\mathcal{P}_i \subseteq \mathcal{P}_j$, and $\mathcal{P}_j \subseteq \mathcal{P}_i,\quad \forall i, j \in \{1, 2, 3\}, i \neq j$) with both our method (Corollary~\ref{cor:alpha_adaptation}) and by calling the \textsc{contains} function of CORA toolbox \cite{Althoff2015ARCH}. Our method validates the set inclusion iff \eqref{eq:cpz_inclusion_a}--\eqref{eq:cpz_inclusion_d} combined with \eqref{eq:alpha_gamma}--\eqref{eq:alpha_psi} returns a feasible result, as implemented with the help of the Yalmip toolbox \cite{lofberg2004yalmip}. CORA successively over-approximates the CPZ sets by polynomial zonotopic sets for which then checks the inclusion using \cite[Proposition 3.1.34]{Kochdumper}. 

For illustration, we depict in the next table the ``ground truth'', the results (inclusion validation), and computation time. Whenever the test function returns that the inclusion holds we put ``1'', and ``0'' otherwise.
\begin{center}
\renewcommand{\arraystretch}{1.25}
\begin{tabular}{c|c|c|c|c|c}
\hline
\multirow{2}{*}{case}& \multirow{2}{*}{Ground truth}& \multicolumn{2}{c|}{CORA}&\multicolumn{2}{c}{Corollary \ref{cor:alpha_adaptation}} \\ \cline{3-6}
&&test&time$[s]$&test&time$[s]$\\\hline
$\mathcal P_1\subseteq \mathcal P_2$& $\checkmark$ &0& 9.15  &1& 1.34\\\hline
$\mathcal P_2\subseteq \mathcal P_1$& $\times$ &0& 10.97 &0& 0.60\\\hline
$\mathcal P_1\subseteq \mathcal P_3$& $\checkmark$ &0& 10.55  &1& 0.83\\\hline
$\mathcal P_3\subseteq \mathcal P_1$& $\times$ &0& 9.60 &0& 0.58\\\hline
$\mathcal P_2\subseteq \mathcal P_3$& $\checkmark$ &0& 10.18  &1& 0.66\\\hline
$\mathcal P_3\subseteq \mathcal P_2$& $\times$ &0& 10.73  &0& 0.81\\
\end{tabular}
\end{center}
We observed reasonable computation times for our implementation and, most importantly, that our method has no false positives or negatives. The CORA method on the other hand, always fails to return a true positive. To the best of our understanding, this happens due to the pre-processing steps which inflate the sets later considered for inclusion testing. Moreover, within CORA, we noticed that the main sub-routine, \textsc{contract} \cite[Proposition 3.1.34]{Kochdumper}, of the \textsc{contains} method, takes roughly $3$ seconds for these examples, the rest of the time being spent in the prepossessing phase.

Simulations were carried out under Linux operating system, and the hardware platform has a 6 cores, 2.70GHz CPU, and 16GB of RAM.

The MATLAB implementation of Corollary~\ref{cor:alpha_adaptation} is stored in the sub-directory ``/replan-public/maximal-positive-invariant-set/ECC-2026'' of the Gitlab repository found at: \url{https://gitlab.com/replan/replan-public.git}.

%% file: chapters/5_conclusions.tex
\section{Conclusions}
\label{sec:conclusions}
In this paper we formulated sufficient inclusion conditions for the CPZ sets, which were later adapted to avoid using the absolute value operator, making the test run efficiently. The simulation for the inclusion test shows promising results, even compared with state-of-the-art toolboxes. Further work will use the inclusion for computing the maximal positive invariant set, in the constrained polynomial zonotopic case.